\documentclass[12pt]{article}
\usepackage{geometry} 
\usepackage{amsfonts}
\usepackage{amsmath}
\usepackage{amsxtra}
\usepackage{tipa}
\usepackage{graphicx}
\usepackage{fullpage}

\newtheorem{theorem}{Theorem}

\newtheorem{proposition}[theorem]{Proposition}

\newenvironment{proof}[1][Proof]{\begin{trivlist}
\item[\hskip \labelsep {\bfseries #1}]}{\end{trivlist}}

\begin{document}

\vspace{4cm}
\noindent
{\bf \Large A Fuchsian matrix differential equation for Selberg correlation integrals}

\vspace{5mm}
\noindent
Peter J.~Forrester${}^*$
 and Eric M.~Rains${}^\dagger$

\noindent
${}^*$Department of Mathematics and Statistics,
University of Melbourne, \\
Victoria 3010, Australia ; \\
${}^\dagger$
Department of Mathematics, California Institute of Technology, Pasadena, CA 91125, USA

\small
\begin{quote}
We characterize averages of $\prod_{l=1}^N|x - t_l|^{\alpha - 1}$ with respect to the Selberg density,
further contrained so that $t_l \in [0,x]$ $(l=1,\dots,q)$ and $t_l \in [x,1]$ $(l=q+1,\dots,N)$, in terms
of a basis of solutions of a particular Fuchsian matrix differential equation. By making use of the Dotsenko-Fateev integrals, the explicit form of the connection matrix from the Frobenius type power series basis to this basis is calculated, thus allowing us to explicitly compute coefficients in the power series expansion of the averages. From these we are able to compute power series for the marginal distributions of the $t_j$ $(j=1,\dots,N)$. In the case $q=0$ and $\alpha < 1$ we compute the
explicit leading order term in the $x \to 0$ asymptotic expansion, which is of interest to the study of an effect known as singularity dominated strong fluctuations. In the case $q=0$ and $\alpha  \in \mathbb Z^+$, and with
the absolute values removed, the average is a polynomial, and we demonstrate that its zeros are
highly structured.
\end{quote}

\section{Introduction}
The Selberg density refers to the probability density function (PDF)
\begin{equation} \label{S}
S_N(\lambda_1,\lambda_2,\lambda;\mathbf{t}):=\frac{1}{S_N(\lambda_1,\lambda_2,\lambda)}\prod_{l=1}^Nt_l^{\lambda_1}(1-t_l)^{\lambda_2}\prod_{1\leq j<k \leq N}|t_k-t_j|^{2\lambda},
\end{equation}
supported on $\mathbf{t}\in[0,1]^n$, where $S_N$ denotes the Selberg integral
\begin{eqnarray}
S_N(\lambda_1,\lambda_2,\lambda)&:=&\int_{[0,1]^n}\prod_{l=1}^Nt_l^{\lambda_1}(1-t_l)^{\lambda_2}\prod_{1\leq j<k \leq N}|t_k-t_j|^{2\lambda} \nonumber \\
& = & \prod_{j=0}^{N-1} {\Gamma(\lambda_1 + 1 + j \lambda) \Gamma(\lambda_2 + 1 + j \lambda) 
\Gamma(1 + (j+1)\lambda) \over \Gamma(\lambda_1 + \lambda_2 + 2 + (N+j-1)\lambda)
\Gamma(1 + \lambda)}.\label{2a}
\end{eqnarray}
(see \cite{FW07p} for a recent review).

Special cases of (\ref{S}), restricted to $2\lambda\in\{1,2,4\}$ and $\lambda_1,\lambda_2$ integers or half integers, arise in classical random matrix theory. For example, consider an $N\times N$ random real orthogonal matrix. With $A$ denoting an $n_1 \times n_2$ sub-block of elements ($n_1\geq n_2, N\geq n_1+n_2$) the eigenvalue PDF for $A^TA$ is given by (\ref{S}) with 
\begin{equation}\label{1}
\lambda_1=\frac{1}{2}(n_1-n_2-1),\hspace{1cm}\lambda_2=\frac{1}{2}(N-n_1-n_2-1),\hspace{1cm}2\lambda=1
\end{equation} 
(see e.g.~\cite[pg.124]{Fo10} ). For general parameters such that (\ref{S}) is well defined, it can be realized as the zero distribution of certain polynomials defined by random three term recurrences \cite{KN04,FR02b,ES06a}, or equivalently as the eigenvalue distribution for certain matrices constructed from of order $N$ (not of order $N^2$) real numbers. For definiteness, we will interpret the variables $\mathbf{t}$ in (\ref{S}) as eigenvalues.

The integral 
\begin{align}
E_N(n;(0,x)):=&\binom{N}{n}\int_{[0,x]^n}dt_1\ldots dt_n\int_{[x,1]^{N-n}}dt_{n+1}\ldots dt_N \,
S_N(\lambda_1,\lambda_2,\lambda;\mathbf{t}) \label{E}
\end{align}
gives the probability that exactly $n$ eigenvalues are contained in the interval $(0,x).$ Let $p_N(n;x)$ denote the PDF for the event that there is an eigenvalue at $x$, with exactly $n$ eigenvalues to the left. Thus $p_N(n;x)$ denotes the PDF for the $(n+1)$-st smallest eigenvalue. It relates to the probabilities (\ref{E}) via the formula
\begin{equation}\label{2}
p_N(n;x)=-\sum_{k=0}^n\frac{d}{dx}E_N(k;(0,x)).
\end{equation}

For its relevance to multivariate statistics, the task of computing the PDFs $p_N(n;x)$ in the case of the parameters (\ref{1}) was undertaken some years ago by Davis \cite{Da72}. His strategy consisted of three steps.
\begin{enumerate}
  \item Establish that $\{E_N(n;(0,x))\}_{n=0,1,\ldots,N}$ occur as part of the matrix of fundamental solutions of a certain Fuchsian matrix differential equation.
   \item To (attempt) to determine the connection matrix transforming the matrix of fundamental solutions implied by the Frobenius type power series solution, both about $x=0$ and $x=1$, to the matrix of fundamental solutions involving (\ref{E}).
   \item To use the differential equation to generate the Frobenius type solutions, and then from step $2$. use the connection matrices to form the appropriate linear combination required to compute the series expansions of $\{E_N(n;(0,x))\}_{n=0,1,\ldots,N}$.
\end{enumerate}

Because step $2.$~was not successfully carried out, the strategy was never executed in its entirety. Here we will return to this problem, and execute all three steps as required for the generation of (\ref{E}) as a power series about $x=0$ and $x=1$, and thus the corresponding power series of (\ref{2}).

The multidimensional integrals generalizing (\ref{E}) by the inclusion of a factor $\prod_{l=1}^N|x-t_l|^{2\mu}$ can also be characterized as solutions of a matrix Fuchsian differential equation. In particular, with $\langle \cdot \rangle$ denoting an average with respect to the Selberg density (\ref{S}), this allows us to calculate the series expansion of the averages of moments of characteristic polynomials
\begin{equation} \label{8.1}
\Big\langle \prod_{l=1}^N |x-t_l|^{2\mu} \Big\rangle.
\end{equation}
Two applications then present themselves. The first relates to the case $2 \mu \in \mathbb Z^+$, when
(\ref{8.1}) is a polynomial. As detailed in Section 3.1 it is known that the zeros of closely related
polynomials are highly structured. We use our ability to compute the polynomial cases of (\ref{8.1}) to (numerically) compute the corresponding zeros for a wide range of the parameters
$N,\lambda_1,\lambda_2,\lambda$. At a qualitative level  families of curves for the loci of zeros are
identified, and for large $\lambda$ we provide an  approximation to the real part of the zeros. 

The second application relates to the series expansion of (\ref{8.1}) with $\mu<0$ as $x\rightarrow1^+$, interpreted as an asymptotic expansion. This is a problem which has attracted previous attention. 
It turns out that then the (near) degeneracies in the spectrum of varying order are being probed, and one has an effect known as singularity dominated strong fluctuations \cite{BK02}. The latter has many applications in addition to the one in number theory discussed in \cite{BK02}. For example, it applies to the analysis of twinkling starlight \cite{Be77}, van Hove-type singularities \cite{Be82}, and the influence of classical periodic orbit bifurcations on quantum energy level \cite{BKS00} and wave function statistics \cite{KP01}.

Let the notation $A(x)\approx B(x)$ mean that there exists two positive numbers $C_1$ and $C_2$ independent of $x$ such that \begin{equation}\label{6a}
C_1 \leq \frac{A(x)}{B(x)}\leq C_2.
\end{equation}
In \cite{FK04} (\ref{8.1}) was expressed as a certain generalized hypergeometric function based on Jack polynomials (see e.g. \cite[Ch. 13]{Fo10}). For the latter Yan \cite{Ya92} had obtained estimates  of the form (\ref{6a}). Consequently it was shown in \cite{FK04} (with $b\mapsto b-1$ to account for the fact that the correction $a\mapsto a+1$ is required in the second argument of ${_2F}_1^{(\beta)}$ in (\ref{r1})) that for $\mu<0$ and $x>1$, with $|\mu|$ taking on the explicit value
\begin{equation}\label{9.a}
2|\mu|=2 \lambda(j-1)+1+\lambda_2, \hspace{1.5cm} j\in \{1,\ldots, N \},
\end{equation}
one has 
\begin{equation}\label{9.1}
x^{-2\mu N}\Big\langle\prod_{l=1}^N|x-x_l|^{2\mu} \Big\rangle\approx(1-1/x)^{-(j-1)j\lambda}\log\frac{1}{1-1/x}.
\end{equation}
Furthermore, for $\mu<0$, $x>1$ and with $|\mu|$ such that 
\begin{equation}\label{9.b}
2\lambda(j-1)+\lambda_2+1<2|\mu| <2\lambda j + \lambda_2+1 \hspace{1.5cm} j\in \{1,\ldots, N \},
\end{equation}
it was shown
\begin{equation}\label{9.2}
x^{-2\mu N} \Big\langle \prod_{l=1}^N |x-x_l|^{2\mu} \Big\rangle \approx (1-1/x)^{-j(2|\mu|-\lambda_2-1-(j-1)\lambda)}.
\end{equation}
Note that (\ref{9.b}) implies that $j$ on the RHS of (\ref{9.2}) can be written 
\begin{equation}\label{9.c}
j=\rm{int}[((2|\mu|-\lambda_2-1)/2\lambda)+1], \hspace{1.5cm} (2|\mu|-\lambda_2-1)/2\lambda+1 \not\in\mathbb{Z},
\end{equation}
where int$[\cdot]$ denotes the integer part.

The present study allows the results (\ref{9.1}) and (\ref{9.2}) to be supplemented by the specification of the series expansion about $x=1$, giving in particular the proportionalities in the leading asymptotic forms.

\section{The Fuchsian matrix differential equation}
\subsection{Specification}

Our aim in this Section is to characterize the family of integrals
\begin{equation}\label{ea}
I_q^{(\alpha)}(x):=\binom{N}{q}\int_{0}^xdt_1\ldots\int_0^xdt_{q}\int_x^1dt_{q+1}\ldots\int_x^1dt_NF(t_1,\ldots,t_N;\alpha;x),
\end{equation}
where
\begin{equation*}
F(t_1,\ldots,t_N;\alpha;x):=\prod_{l=1}^N t_l^{\lambda_1}(1-t_l)^{\lambda_2}|x-t_l|^{\alpha-1}\prod_{1\leq j<k\leq N}|t_k-t_j|^{2\lambda}
\end{equation*}
in terms of the solution of a Fuchsian matrix differential equation. 
The integrals are well defined for
\begin{equation}\label{ea1}
{\rm Re} \, \lambda_1 > -1, \quad {\rm Re} \, \lambda_2 > -1, \quad {\rm Re} \, \alpha > 0,
\quad {\rm Re} \, \lambda > 0.
\end{equation}
In the case $\alpha=1,$ $2\lambda=1,$ which is recalled in the introduction has motivation in mathematical statistics, this was first accomplished by Davis \cite{Da72}.

Introduce the generalization of (\ref{ea})
\begin{align}
J_{p,q}^{(\alpha)}(x):=&\int_0^xdt_1\ldots \int_0^xdt_{q}\int_x^1dt_{q+1}\ldots \int_x^1dt_N F(t_1,\ldots,t_N;\alpha;x) \notag \\
&\hspace{1cm}\times e_p(t_1-x,\dots,t_N-x) \label{eb}
\end{align}
where $e_p$ denotes the elementary symmetric function
\begin{equation*}
e_p(x_1,\ldots,x_N):=\sum_{1\leq j_1 <j_2 < \ldots < j_p \leq N} x_{j_1}x_{j_2}\ldots x_{j_N}.
\end{equation*}
It was shown in \cite{Fo93} that for general parameters (\ref{eb}) satisfies the differential-difference equation 
\begin{align}
&(N-p)E_p J_{p+1,q}^{(\alpha)}(x)= \notag \\
&\hspace{2cm} -(A_px+B_p) J_{p,q}^{(\alpha)}(x)+x(x-1)\frac{d}{dx}J_{p,q}^{(\alpha)}(x)+D_px(x-1)J_{p-1,q}^{(\alpha)}(x) \label{ec}
\end{align}
where
\begin{align}
A_p&=(N-p)(\lambda_1+\lambda_2+2\lambda(N-p-1)+2\alpha) \notag\\
B_p&=(p-N)(\lambda_1+\lambda(N-p-1)+\alpha) \notag\\
D_p&=p(\lambda(N-p)+\alpha) \notag\\
E_p&=\lambda_1+\lambda_2+\lambda(2N-p-2)+\alpha +1 \label{ed}.
\end{align}
We see from (\ref{ed}) that the equation (\ref{ec}) is independent of $q$; in fact the derivation of \cite{Fo93} was carried out with $q=N$. However the same derivation ---  which relies on integration by parts --- applies for general $q=0,1,\ldots,N.$

Let us fix $q$ and form the column vector
\begin{equation} \label{17a}
\mathbf{J}_q^{(\alpha)}(x)=[J_{p,q}^{(\alpha)}(x)]_{p=0,1,\ldots,N},
\end{equation}
It follows immediately from (\ref{ec}) that $\mathbf{J}_q^{(\alpha)}$ satisfies a particular Fuchsian matrix differential equation. This in turn (as noted by Davis \cite{Da72} in the special case of that work) assumes a simpler form if we set 
\begin{equation}\label{JH}
J_{p,q}^{(\alpha)}(x) = (x-1)^p H_{p,q}^{(\alpha)}(x), \qquad \mathbf{H}_q^{(\alpha)}(x) =
 [H_{p,q}^{(\alpha)}(x)]_{p=0,1,\dots,N}.
 \end{equation}
 
\begin{proposition}\label{2.1}
We have
\begin{equation} \label{ee}
\frac{d}{dx}\mathbf{H}_q^{(\alpha)}=\bigg(\frac{Y^+}{x}+\frac{Y^{-}}{1-x}\bigg)\mathbf{H}_q^{(\alpha)}(x)
\end{equation}
where
\[ Y^{+}=\left[ \begin{array}{ccccc}
\sigma_N & a_{N-1} & 0 & \cdots & 0 \\
0 & \sigma_{N-1} & a_{N-2} & \cdots & \vdots \\
0 &  \ddots & \ddots & \ddots & 0 \\
\vdots &&&\sigma_{1}&a_{0} \\
0&\cdots&0&0&\sigma_0
\end{array} \right] \] 
\[ Y^{-}=\left[ \begin{array}{ccccc}
b_N & 0 & 0 & \cdots & 0 \\
c_N & b_{N-1} & 0 &  & \vdots \\
0 &  \ddots & \ddots & \ddots & 0 \\
\vdots &&c_{2}&b_{1}& 0 \\
0&\cdots&0&c_1&b_0
\end{array} \right] \] 
with
\begin{align}
\sigma_i&=-B_{N-i}=i(\lambda_1+\lambda(i-1)+\alpha) \notag \\
a_i&=-iE_{N-i}=-i(\lambda_1+\lambda_2+\lambda(N+i-2)+\alpha +1) \notag \\
b_i&=-(A_{N-i}+B_{N-i})+N-i=-i(\lambda_2+\lambda(i-1)+\alpha)+N-i \notag\\
c_i&= D_{N-i}= (N-i)(\lambda i+\alpha) \label{19a}
\end{align}
\end{proposition}

\noindent
{\bf Remarks.}

\begin{enumerate}
\item In terms of $\mathbf{J}_q^{(\alpha)}(x)$, the structure of the matrices on the RHS of
(\ref{ee}) alters and the number of matrices increases from two to three.
\item Suppose that instead of expanding about $x=0$ we wanted to expand about $x=1$.
With $y=1-x$ and $\overline{\mathbf{H}}_q^{(\alpha)}(y) := \mathbf{H}_q^{(\alpha)}(x) $ we see from
(\ref{ee}) that
\begin{equation}\label{Hq}
- {d \over dy} \overline{\mathbf{H}}_q^{(\alpha)}(y)  = \Big (
{1 \over y} Y^- + {1 \over 1 - y} Y^+ \Big )  \overline{\mathbf{H}}_q^{(\alpha)}(y) .
 \end{equation}
 Moreover, it follows from (\ref{eb}) that
\begin{equation}\label{Hql} 
J_{p,q}^{(\alpha)}(1-x) =(-1)^p J_{p,N-q}^{(\alpha)}(x) \Big |_{\lambda_1 \leftrightarrow \lambda_2}.
 \end{equation}
 \item In \cite{Mi07} the integrals (\ref{eb}) with $e_p(t_1-x,\dots,t_N-x)$ replaced by
 $$
 \prod_{j=1}^N (t_j - 1)^{-1} e_p ( 1 - 1/t_1,\dots, 1 - 1/t_p)
 $$
 have been considered, and a recurrence corresponding to a matrix differential system with a structure
 identical to that exhibited by (\ref{ee}) obtained.
 \item By examining the $x\rightarrow0$ behaviour of the integrals in (\ref{17a}) (see Proposition \ref{prop2} below) it can be established that
 \begin{equation*}
J(x):={\rm det} [ \mathbf{J}_q^{(\alpha)}(x)]_{q=0,\ldots,N}\not=0.
 \end{equation*}
 This was noted by Davis \cite{Da72} in the special case of that work. Moreover, according to (\ref{JH})
 \begin{equation*}
J(x):=(x-1)^{N(N+1)/2} {H}(x),\hspace{1.5cm} H(x):={\rm det} [ \mathbf{H}_q^{(\alpha)}(x)]_{q=0,\ldots,N}
 \end{equation*}
 while it follows from (\ref{ee}) that
 \begin{equation*}
{H}'(x)=\Big(\frac{1}{x}{\rm Tr}\hspace{0.1cm}Y^++\frac{1}{1-x}{\rm Tr} \hspace{0.1cm}Y^- \Big){H}(x).
 \end{equation*}
 Thus
 \begin{equation*}
 {H}(x)= {H}_0 \, x^{{\rm Tr} \hspace{0.05cm}Y^+}(1-x)^{{\rm -Tr} \hspace{0.05cm}Y^-}
  \end{equation*}
 for some ${H}_0$ independent of $x$, which can in fact be made explicit from the $x \rightarrow 0$ asymptotic form. Such determinants of Selberg integrals have been considered in e.g. \cite{Va90a,Va90b}.
 \end{enumerate}

\subsection{Series solutions}

The differential equation (\ref{ee}) is a particular example of a matrix differential system
\begin{equation}\label{WA}
\frac{d}{dx}\mathbf{w}_j(x)=A(x)\mathbf{w}_j(x), \hspace{1.5cm} (j=1,\ldots,N),
\end{equation}
where $A(x)$ has a simple pole at $x=0$ and is analytic for $0<|x|<1$. The theory of such systems is well established \cite[Section 24]{Wa98a}. In particular, one has that there is a fundamental matrix of $N+1$ linearly independent vector solutions of the form
\begin{equation}\label{Ub}
W(x)=U(x)x^B=:U(x)e^{B\log x}.
\end{equation}
Here $U(x)$ is a matrix analytic for $|x|<1$ and so can be expressed as a power series while $B$ is a constant matrix. If $B$ is diagonalizable
\begin{equation}\label{Ux}
e^{B \log x}={\rm diag}(x^{\nu_0},x^{\nu_1},\ldots,x^{\nu_{N}})
\end{equation}
where $\{ \nu_i \}$ are the eigenvalues of $B$. Otherwise $C^{-1}x^{B \log x}C$ has diagonal terms of the form $x^{\nu_i}$ while the upper triangular terms, partitioned into Jordan blocks,
are polynomials in $\log x$ of degree strictly less than the row size of the corresponding Jordan block.

In practice, unless (\ref{Ux}) holds, the decomposition (\ref{Ub}) is not accessed directly from (\ref{WA}). Instead one seeks solutions of the form
\begin{equation} \label{WA1}
\mathbf{w}_k(x)=x^{\sigma_k}\sum_{l=0}^\infty \mathbf{p}_l(\log x) x^l,
\end{equation}
where the $\mathbf {p}_l(\log z)$ --- which also depend on $k$ ---  are polynomials in $\log z$ of degree bounded as for $B$ and $\{\sigma_k\}_{k=0,1\dots,N}$ are the eigenvalues of
$A_1 := \lim_{x \to 0} x A(x)$. Furthermore the $\{\nu_j\}$ in (\ref{Ux}) and the $\{\sigma_k\}$ must be
equal modulo an integer.

Suppose in (\ref{WA}) we write
\begin{equation*}
A(x)=\frac{1}{x}\sum_{j=0}^\infty A_jx^j.
\end{equation*}
For $\sigma$ an eigenvalue of $A_0$, and with $s=\log x$, it is straightforward to show that the polynomials in (\ref{WA1}) are determined by the recurrences
\begin{align}
\mathbf{p}_0'+(\sigma \mathbb{I}-A_0)\mathbf{p}_0&=\mathbf{0} \notag \\
\mathbf{p}_l'+((\sigma+l)\mathbb{I}-A_0)\mathbf{p}_l&=\sum_{i=0}^{l-1}A_{l-i}\mathbf{p}_i \hspace{1cm} (l=1,2,\ldots) \label{89}.
\end{align}
The simplest situation for the solution of these recurrences is when $\mathbf{p}_l'=0$ for each $l=0,1,\ldots$, and thus there are no $\log$ terms. For this to occur it is sufficient that 
\renewcommand{\labelenumi}{(\roman{enumi})}
\begin{enumerate}
\item $A_0$ is diagonalizable; 
\item no two eigenvalues of $A_0$ differ by an integer\label{90}.
\end{enumerate}
Of these only (i) is necessary, as it possible that a pair of eigenvalues of $A_0$ differ by an integer yet there is no $\log$ terms. The latter scenario occurs exactly when for some $l$ the matrix $(\sigma+l)\mathbb{I}-A_0$ is singular, the $\mathbf{p}_i$ for $i<l$ are all constant vectors, and $\sum_{i=0}^{l-1}A_{l-i}\mathbf{p}_i$ is in the column space of $(\sigma+ l)\mathbb{I}-A_0$. Then the second equation in (\ref{89}) still permits a solution with $\mathbf{p}_l'=0$. Moreover the solution contains an arbitrary parameter corresponding to a scalar multiple of a vector from the null space, implying that the eigenvalues $\sigma$ and $\sigma+l$ of $A_0$ give linearly independent series solutions. 

In the situation that conditions (i) and (ii) hold we have $\mathbf{p}_0$ is an eigenvector of $A_0$ corresponding to the eigenvalue $\sigma_k$. The vectors $\{\mathbf{p}_l\}_{l=1,2,\ldots}$ are then given by the recurrence 
\begin{equation}\label{26a}
\mathbf{p}_l=((\sigma_k+l)\mathbb{I}-A_0)^{-1}\sum_{i=0}^{l-1}A_{l-i}\mathbf{p}_i.
\end{equation}
Furthermore, when condition (ii) does not hold but there are no $\log$ terms, the recurrence (\ref{26a}) can be used if we first perturb the parameter $\lambda_1$ (say) so that $(\sigma_k+l)\mathbf{I}-A_0$ is no longer singular. After $p_l$ has been computed, we can then take the limit that $\lambda_1$ returns to its original value.

\subsection{Connection matrix}
Let $\mathbf{w}_k(x)=:[(\mathbf{w}_k(x))_j]_{j=0,\ldots,N}$ be the vector of power series solutions of (\ref{ee}) corresponding to the eigenvalue $\sigma_k$ of $Y^+$. Let
\begin{equation}\label{28'}
\mathbf{H}_q^{(\alpha)}(x)=[(x-1)^{-p}J_{p,q}^{(\alpha)}(x)]_{p=0,\ldots,N}
\end{equation}
be the vector formed from the integrals (\ref{eb}) and also satisfying (\ref{ee}). Form the fundamental matrices $W(x)=[\mathbf{w}_k(x)]_{k=0,\ldots,N}$ and $H^{(\alpha)}(x)=[\mathbf{H}_q^{(\alpha)}(x)]_{q=0,\ldots,N}$. Since the system (\ref{ee}) has dimension $N+1$, and $\{\mathbf{w}_k(x)\}_{k=0,\ldots,N}$ and $\{\mathbf{H}_q(x)\}_{q=0,\ldots,N}$ are bases, there exists a connection matrix $C=[c_{ij}]_{i,j=0,\ldots,N}$ independent of $x$ such that
\begin{equation}\label{HWS}
H^{(\alpha)}(x)=W(x)C.
\end{equation} 
We are particularly interested in the integrals in the first row of $H^{(\alpha)}(x),$ $J_{0,q}^{(\alpha)}=I_q^{(\alpha)}(x)$, for which (\ref{HWS}) gives

\begin{equation}\label{HWS1}
[I_q^{(\alpha)}(x)]_{q=0,\ldots,N}=[(\mathbf{w}_k(x))_0]_{k=0,\ldots,N}\hspace{0.1cm}C
\end{equation}

With the power series of $(\mathbf{w}_k(x))_0$ calculated according to the strategy of the previous section, knowledge of $C$ will enable us to compute the power series of $I_q^{(\alpha)}(x)$. Conversely, knowledge of the small $x$ expansion of $I_q^{(\alpha)}(x)$, deduced in certain restricted regions of parameter space, can be used to compute the entries of $C$. For this purpose use will be made of the Selberg integral (\ref{2a}), as well as a variant of the Selberg integral due to Dotsenko and Fateev \cite{DF85}
\begin{align}
S_{(p,N-p)}(\lambda_1,\lambda_2,\lambda):=&\int_0^1dt_1\ldots \int_0^1dt_p\int_1^\infty dt_{p+1}\ldots\int_1^\infty dt_N \notag \\
&\times \prod_{l=1}^Nt_l^{\lambda_1}|t_l-1|^{\lambda_2}\prod_{1\leq j<k\leq N}|t_k-t_j|^{2\lambda}.\label{SDF}
\end{align}
As written, this converges for $|\lambda|\ll1$ and $-1<{\rm Re}(\lambda_1)\ll0$, $-1<{\rm Re}(\lambda_2)\ll0$ and is defined by analytic continuation outside this region.
According to \cite[Prop. 4.5.1]{Fo10} these integrals satisfy the first order recurrence
\begin{align*}
S_{(p,N-p)}(\lambda_1,\lambda_2,\lambda)=&\frac{p}{N-p+1}\frac{\sin \pi( N-p+1)\lambda \sin \pi(\lambda_1+\lambda_2+2+(N+p-2)\lambda)}{\sin\pi p \lambda \sin\pi(\lambda_1+1+(p-1)\lambda)}\\
&\hspace{1cm}\times S_{(p-1,N-p+1)}(\lambda_1,\lambda_2,\lambda),
\end{align*}
and it follows from this that (\ref{SDF}) can be expressed in terms of the Selberg integral according to 
\begin{align}
S_{(p,N-p)}(\lambda_1,\lambda_2,\lambda)=&S_N(\lambda_1,\lambda_2,\lambda)\prod_{j=1}^{N-p}\frac{j}{N-j+1} \notag \\
&\hspace{1cm} \times \frac{\sin \pi (N-j+1)\lambda \sin \pi (\lambda_1+1+(N-j)\lambda)}{\sin \pi j \lambda \sin \pi (\lambda_1+\lambda_2+2+(2N-j-1)\lambda)} \label{29a}.
\end{align}
\begin{proposition}\label{prop2}
Let $\lambda_1,\lambda_2,\alpha,\lambda$ be constrained according to (\ref{ea1}) so that (\ref{ea}) is well defined. For 
\begin{equation}\label{r1}
{\rm Re} (\lambda_1+\alpha+\lambda(q-1) )>0 
\end{equation}
we have
\begin{equation}\label{ep1}
I_q^{(\alpha)}(x)\mathop{\sim}_{x\rightarrow0^+}x^{\sigma_q}S_q(\lambda_1,\alpha-1,\lambda)S_{N-q}(\lambda_1+\alpha-1+2q\lambda,\lambda_2,\lambda)
\end{equation}
while for
\begin{equation}\label{r2}
-2l {\rm Re}\lambda<{\rm Re}(\lambda_1+\alpha)<-2(l-1){\rm Re}\lambda,\hspace{2cm} l>q
\end{equation}
we have
\begin{align}
I_q^{(\alpha)}(x)\mathop{\sim}_{x \rightarrow 0^+}&x^{\sigma_l}\binom{N-q}{l-q} S_{N-l}(\lambda_1+\alpha-1+2l\lambda,\lambda_2,\lambda) \notag \\
&\hspace{1cm} \times S_{(q,l-q)}(\lambda_1,\alpha-1,\lambda).\label{ep2}
\end{align}
\end{proposition}
\begin{proof}
Let $\sigma(x):=x(\lambda_1+\alpha+\lambda(x-1))$ so that $\sigma(k)=\sigma_k$, $k=0,1,\ldots,N$. The significance of (\ref{r1}) is that for this range of parameters (assumed real)
\begin{equation}\label{1q}
\sigma_q>\sigma_j \hspace{1cm} (j=0,\ldots,q-1),\hspace{1cm}\sigma_q<\sigma_j \hspace{1cm}(j=q+1,\ldots,N).
\end{equation}
To see this, we note from the form of $\sigma(x)$ that (\ref{1q}) is equivalent to requiring $\sigma_q>0$ and thus $\lambda_1+\alpha+\lambda(q-1)>0$.
In regards to (\ref{ep1}), for $j=1,\ldots,q$ we change variables $t_j\mapsto xt_j$ to obtain
\begin{align*}
I_q^{(\alpha)}(x)=&x^{\sigma_q}\binom{N}{q}\int_0^1dt_1\ldots \int_0^1dt_q\prod_{l=1}^qt_l^{\lambda_1}(1-xt_l)^{\lambda_2}|1-t_l|^{\alpha-1}\prod_{1\leq j<k\leq q}|t_k-t_j|^{2\lambda}\\
&\times \int_x^1dt_{q+1}\ldots\int_x^1dt_N\prod_{l=q+1}^Nt_l^{\lambda_1}(1-t_l)^{\lambda_2}|x-t_l|^{\alpha-1}\prod_{q+1\leq j<k \leq N}|t_k-t_j|^{2\lambda} \\
&\times\prod_{j=1}^q\prod_{k=q+1}^N|xt_j-t_k|^{2\lambda}.
\end{align*}
Setting $x=0$ in the integrands and comparing with (\ref{2a}) gives (\ref{ep1}).

The significance of (\ref{r2}) is that with real parameters
\begin{equation}\label{ss}
\sigma_l<\sigma_j \hspace{1cm} (j=1,\ldots, N) \hspace{1cm} j \not= l.
\end{equation}
To see this we note that for (\ref{ss}) to hold we require
\begin{equation}\label{ts}
l-\frac{1}{2}<x^*<l+\frac{1}{2}
\end{equation}
where $x^*$ is the maximum of $\sigma(x)$ and so 
\begin{equation}\label{us}
x^*=-\frac{\lambda_1+\alpha-\lambda}{2\lambda}.
\end{equation}
Substituting (\ref{us}) in (\ref{ts}) gives (\ref{r2}).
In regards to (\ref{ep2}), the same strategy is employed, except that we rescale the first $q$ coordinates $t_j \mapsto xt_j$ together with $l-q$ of the next $N-q$ coordinates, (this choice accounts for the combinatorial factor in (\ref{ep2})), and compare the resulting integrals with $x=0$ to both (\ref{2a}) and (\ref{SDF}).
\hfill $\square$ \end{proof}

Our strategy to use knowledge of the above asymptotic behaviour to compute the entries of $C$ will require restricting attention to regions of parameter space as specified by inequalities (\ref{r1}) and (\ref{r2}). To then justify that our expressions are valid generally, we require knowledge of the anaalytic structure of the entries of $C$ as a function of $\lambda_1$ in the complex plane.

\begin{proposition}\label{prop2b}
Both matrices of solutions in (\ref{HWS1}) are  meromorphic functions of $\lambda_1$ and so consequently is the connection matrix $C$.
\end{proposition}
\begin{proof}
We see from (\ref{89}) that each coefficient $\mathbf{p}_l$ in $\mathbf{w}_k(x)$ is a rational function of the parameters $\lambda_1,\lambda_2,\lambda,\alpha$ with poles at the zeros of $\det ((\sigma_k+l)\mathbf{I}-A_0)=0$. Furthermore, we know that $x^{-\sigma_k}\mathbf{w}_k(x)$ has radius of convergence $R=1$ about $x=0$. Hence inside this radius of convergence $x^{-\sigma_k}\mathbf{w}_k(x)$ is a meromorphic function of the parameters.

On the other hand, consider the integrals (\ref{eb}). As written they are analytic functions of the parameters in the region (\ref{ea1}). By suitable ordering of the integration variables
\begin{equation*}
0<t_1<\ldots<t_q<x<t_{q+1}<\ldots<t_N<1
\end{equation*}
say, and a convention for the meaning of the power functions, the integrand is an analytic function of the integration variables $\mathbf{t}$, and more particularly $t_1,\ldots,t_q$, in the region (\ref{ea1}).

Now let $\gamma$ be a simple contour starting at $x$ and looping once around the origin before returning to $x$. Suppose that the integration variables $t_1,\ldots,t_q$ are confined to $\gamma$ with some prescribed ordering. The resulting integral, $J_{p,q}^{(\alpha)}(x;\gamma)$ say, is analytic in $\lambda_1.$ Moreover, by deforming $\gamma$ to loop around the origin by running along the real axis, we see that $J_{p,q}^{(\alpha)}(x;\gamma)$ is equal to $J_{p,q}^{(\alpha)}(x)$ times a finite sum of certain phases, which are complex exponentials in $\lambda_1$ and $\lambda$. Consequently $J_{p,q}^{(\alpha)}(x)$ itself is a meromorphic function of $\lambda_1$, with poles at the zeros of the sum of phases.
\hfill $\square$ \end{proof}

The series solutions $(\mathbf{w}_q(x))_0$ are unique only up to a scalar factor. We choose the latter so that in the range (\ref{r1})
\begin{equation}\label{wI}
\lim_{x \rightarrow 0}(\mathbf{w}_q(x))_{(0)}/I_{q}^{(\alpha)}(x)=1.
\end{equation}
We can now proceed to use Propositions \ref{2.1} and \ref{prop2} to determine the entries of $C$.
\begin{proposition}\label{prop4}
With $C=[c_{k,q}]_{k,q=0,\ldots,N}$ we have
\begin{align}
c_{k,q}&=0,\hspace{10.3cm}k<q \label{C1}\\
c_{k,q}&=\prod_{j=1}^{k-q}\frac{\sin \pi (k+1-j) \lambda \sin \pi(\lambda_1+(k-j)\lambda)}{\sin \pi j \lambda \sin \pi (\lambda_1+\alpha+(2k-j-1)\lambda)}, \hspace{3.2cm} k\geq q\label{C2}.
\end{align}
In particular, for $k\geq q$ $c_{k,q}$ is periodic under $\lambda_1\mapsto \lambda_1+1$ or $\lambda \mapsto \lambda+1$ and is multiplied by $(-1)^{k-q}$ under $\alpha \mapsto \alpha+1$. Furthermore, for $\lambda$ an integer
\begin{equation} \label{9a}
c_{k,q}=c_{k,q}\Big |_{\lambda=0}=\binom{k}{q}\prod_{j=1}^{k-q}\frac{\sin \pi \lambda_1}{\sin \pi (\lambda_1+\alpha)}
\end{equation}
while for $\lambda=r/s$, $r$ and $s$ relatively prime positive integers
\begin{equation} \label{9b}
c_{k,q}\Big |_{\lambda=r/s}=c_{k\hspace{0.05cm}{\rm mod}\hspace{0.05cm}s,q\hspace{0.05cm}{\rm mod} s}\Big |_{\lambda=r/s} c_{\lfloor k/s \rfloor,\lfloor q/s \rfloor}\Big |_{\substack{\lambda_1 \mapsto s\lambda_1,\alpha\mapsto s \alpha\\ \lambda=0}}. 
\end{equation}
\end{proposition}
\begin{proof}
Suppose the parameters are such that (\ref{r1}) is valid. Then from (\ref{1q}) $\sigma_k<\sigma_q$ for $k=0,1,\ldots,q-1$. Taking the limit $x \rightarrow 0$ in the connection formula
\begin{equation}\label{CNF}
\sum_{k=0}^{N}(\mathbf{w}_k(x))_0c_{k,q}=I_q^{(\alpha)}(x)
\end{equation}
it follows that (\ref{C1}) must hold, and using too (\ref{wI}), that
\begin{equation*}
c_{q,q}=1.
\end{equation*}
Suppose now the parameters are such that (\ref{r2}) is valid. Then (\ref{ss}) holds. Taking $x\rightarrow 0$ in (\ref{CNF}) gives that the LHS is dominated by the term $k=l$ and thus
\begin{equation}\label{51}
c_{l,q}=\lim_{x\rightarrow0}\frac{I_q^{(\alpha)}(x)}{(\mathbf{w}_{l}(x))_0},\hspace{2cm l>q}.
\end{equation}
But according to (\ref{wI}) the coefficient of $x^{\sigma_{l}}$ in $(\mathbf{w}_{l}(x))_0$ is the same as the coefficient of $x^{\sigma_{l}}$ in $I_l^{(\alpha)}(x)$ in the range (\ref{r1}). Thus, using (\ref{ep2}) in the numerator and (\ref{ep1}) in the denominator, it follows from (\ref{51}) that for $l>q$ 
\begin{equation*}
c_{l,q}=\frac{\binom{N-q}{l-q}\binom{N}{q}}{\binom{N}{l}}\frac{S_{(q,l-q)}(\lambda_1,\alpha-1,\lambda)}{S_l(\lambda,\alpha-1,\lambda)}.
\end{equation*} 
Now using (\ref{29a}) gives (\ref{C2}).
\hfill $\square$  \end{proof}

Let $(\mathbf{w}_{k}(x))_0^*$ denote $(\mathbf{w}_{k}(x))_0$ normalized so that the coefficient of $x^{\sigma_k}$ is unity, and thus
\begin{equation}\label{44'}
(\mathbf{w}_k(x))_0^*=x^{\sigma_k}(1+\sum_{l=1}^\infty p_{l.k}x^k)
\end{equation}
where the coefficients $p_{l,k}$ are independent of $x$ (the structure (\ref{44'}) holds for the generic solution if no $\log$ terms are present; as will be seen in Section 3.2 cases with $\log$ terms can be obtained as limits of the generic case). It then follows from (\ref{wI}), (\ref{ep1}), (\ref{CNF}) and (\ref{C1}) that 
\begin{equation}
I_q^{(\alpha)}(x) =\sum_{k=q}^{N}(\mathbf{w}_{k}(x))_0^*S_{k}(\lambda_1,\alpha-1,\lambda)S_{N-k}(\lambda_1+\alpha-1+2k\lambda,\lambda_2,\lambda) \binom{N}{k}c_{k,q}\label{ISC}.
\end{equation}

\noindent
{\bf Remark.}

Let
\begin{align*}
K_{p,q}^{(\alpha)}(x)=&(x-1)^p\int_{-\infty}^0dt_1\ldots \int_{-\infty}^0dt_q\int_x^1dt_{q+1}\ldots \int_x^1dt_N \\
&\times F(t_1,\ldots,t_N;\alpha;x)e_p(t_1-x,\ldots,t_N-x)
\end{align*}
(cf. (\ref{eb}) and (\ref{JH})), and write $K^{(\alpha)}(x)=[K_{p,q}^{\alpha}(x)]_{p,q=0,\ldots,N}$. In \cite{Mi07} the connection matrix $C^{(0,1)}$ defined so that
\begin{equation*}
H^{(\alpha)}(x)=K^{(\alpha)}(x)C^{(0,1)}
\end{equation*}
has been computed.

\subsection{Action of the monodromy around $0$}
A basic question relating to the solution matrix of integrals $H^{(\alpha)}(x)$ is how it transforms upon a circuit about the origin, $x\mapsto xe^{2\pi i}$, or a circuit about the other singular points of the Fuchsian equation (\ref{JH}), namely $x=1$ and $x=\infty$. Consider first a circuit about the origin. Thus we seek a matrix $M_0$ --- referred to as the monodromy matrix --- such that
\begin{equation} \label{JM}
H^{(\alpha)}(x)\bigg|_{x\mapsto xe^{2\pi i}}=H^{(\alpha)}(x)M_0. 
\end{equation}

\begin{proposition}
Suppose the series solutions (\ref{WA1}) have no $\log$ terms. The monodromy matrix $M_0$ in (\ref{JM}) is given in terms of the connection matrix $C$ in (\ref{HWS}) by
\begin{equation} \label{JM1}
M_0=C^{-1}DC
\end{equation}
where
\begin{equation}\label{D}
D={\rm diag}[e^{2\pi i q (\lambda_1+ \alpha -1 +(q-1)\lambda)}]_{q=0,\ldots,N}.
\end{equation}
\end{proposition}
\begin{proof}
With $W(x)$ specified as below (\ref{28'}), and $\mathbf{p}_l(\log x)=\mathbf{p}_l$ (i.e. no $\log$ terms), we see that
\begin{equation*}
W(x)\bigg|_{x\mapsto x e^{2 \pi i}}=W(x)D.
\end{equation*}
Using this in (\ref{HWS}) shows
\begin{equation*}
H^{(\alpha)}(x)\bigg|_{x\mapsto x e^{2\pi i}}=W(x)DC=H^{(\alpha)}(x)C^{-1}DC.
\end{equation*}
Comparison with (\ref{JM}) gives (\ref{JM1}).
\hfill $\square$ \end{proof}

Consider now how the solution matrix of integrals $H^{(\alpha)}(x)$ transforms upon a circuit about $x=1$. Writing $y=1-x$, $\overline{H}^{(\alpha)}(y)=H^{(\alpha)}(x)$ as is consistent with the use in (\ref{19a}), we seek the monodromy matrix $M_1$ such that
\begin{equation}\label{JT}
\overline{H}^{(\alpha)}(y)\bigg|_{y\mapsto y e^{2\pi i}}=\overline{H}^{(\alpha)}(y)M_1.
\end{equation}
According to (\ref{Hq})
\begin{equation}\label{JS}
\overline{H}^{(\alpha)}(y)=H^{(\alpha)}(x)\bigg|_{\lambda_1 \leftrightarrow \lambda_2}\hspace{0.1cm} \overline{I}
\end{equation}
where $\overline{I}$ is the $(N+1)\times(N+1)$ matrix with $1$'s along its anti-diagonal, and $0$'s elsewhere. Hence
$$
\overline{H}^{(\alpha)}(y)\bigg|_{y\mapsto y e^{2\pi i}} =H^{(\alpha)}(x)\bigg|_{\substack{\lambda_1 \leftrightarrow \lambda_2 \\ x\mapsto xe^{2\pi i}}} \hspace{0.1cm} \overline{I} 
=\Big(H^{(\alpha)}(x)M_0 \Big)\bigg|_{\lambda_1 \leftrightarrow \lambda_2} \hspace{0.1cm} \overline{I}=\overline{H}^{(\alpha)}(y)\hspace{0.1cm}\Big(M_0 \bigg|_{\lambda_1 \leftrightarrow \lambda_2} \Big)\hspace{0.1cm} \overline{I} 
$$
where the second equality follows from (\ref{JM}) and the third from a second application of (\ref{JS}). Comparing with (\ref{JT}) we read off that
\begin{equation}\label{IMI}
M_1=\overline{I}\Big(M_0\bigg|_{\lambda_1 \leftrightarrow \lambda_2} \Big)  \overline{I} .
\end{equation}

For the circuit about $x=\infty$ we set $\widetilde{H}^{(\alpha)}(z):=H^{(\alpha)}(\frac{1}{z})$ and seek the monodromy matrix $M_\infty$ such that 
\begin{equation*}
\widetilde{H}^{(\alpha)}(z)\bigg|_{z\mapsto ze^{2 \pi i}}=\widetilde{H}^{(\alpha)}(z)M_{\infty}.
\end{equation*}
Here no further calculation is necessary due to the fundamental relationship 
\begin{equation}\label{IMII}
M_\infty M_0 M_1 = \mathbb{I}.
\end{equation}
In words this says that the matrix $H^{(\alpha)}(x)$ is single valued when $x$ is traced around a contour starting at $\frac{1}{2}-i \epsilon$ $(0<\epsilon\ll1)$ say, looping around the points $x=1,0$ then $\infty$ while not crossing the interval $[0,1]$ on the real axis and $[0,-i \infty)$ on the imaginary axis.

We have from (\ref{JM1}), (\ref{IMI}), and (\ref{IMII}) that the monodromy matrices are determined by the connection matrix $C$ and the diagonal matrix $D$  (\ref{D}). The latter is periodic under any of the shifts 
$\lambda_1\mapsto \lambda_1+1,$ $\alpha \mapsto \alpha+1$ or $\lambda \mapsto \lambda+1$, and we know from Proposition \ref{prop4} that $C$ has the same periodicity properties (with the qualification that each entry $c_{k,q}$ is multiplied by $(-1)^{k-q}$ under 
$\alpha \mapsto \alpha+1$). 
As a consequence $H^{(\alpha)}(z)\Big|_{\lambda_1\mapsto \lambda_1+1}$ (and similarly for $\alpha\mapsto \alpha+1$ or $\lambda \mapsto \lambda+1$) must satisfy a first order matrix difference equation
\begin{equation}\label{HAA}
H^{(\alpha)}\bigg|_{\lambda_1\mapsto\lambda_1+1}=H^{(\alpha)}(x)A
\end{equation}
for some $(N+1)\times(N+1)$ matrix $A$, dependent on the parameters $\lambda_1,\lambda_2,\alpha,N$ and analytic in $x$. In fact in the recent work \cite{FI10} the explicit form of the matrix $A$, by way of its Gauss decomposition, has been found. It turns out to be a polynomial in $x$. The working of \cite{FI10} also applies to the shift $\alpha \mapsto \alpha+1$, but not $\lambda \mapsto \lambda+1$ for which the explicit form of $A$ remains unknown.

It is a basic fact that successive differentiation of the first order matrix differential equation (\ref{ee}) can be used to obtain a scalar differential equation of degree $N+1$ for entries of the first row. This scalar differential equation then has $N+1$ linearly independent integral solutions $\{I_q^{(\alpha)}(x) \}_{q=0,\ldots,N}$ given by (\ref{ea}). In the case $\lambda=0$ we have
\begin{equation*}
I_q^{(\alpha)}(x)=\binom{N}{q}\Big(\int_0^x t^{\lambda_1}(1-t)^{\lambda_2}(x-t)^{\alpha-1}dt \Big)^q\Big(\int_x^1 t^{\lambda_1}(1-t)^{\lambda_2}(t-x)^{\alpha-1}dt \Big)^{N-q}.
\end{equation*}
In general the scalar differential operator $\mathcal{L}_{N+1}$ of degree $N+1$ that has  $\{ (f_1(x))^q(f_2(x))^{N-q}\}_{q=0,\ldots,N}$ as its basis, where $f_1$ and $f_2$ are the basis of the second order differential operator $\mathcal{L}_2$, is called the $N$-th symmetric power of $\mathcal{L}_2.$ Thus we have that the scalar differential operator corresponding to (\ref{ea}) with $\lambda=0$ is the $N$-th symmetric power of the scalar differential operator in the case $N=1$ (which is independent of $\lambda$). Due to the periodicity in $\lambda$, the connection matrix and monodromy must then correspond to this same $N$-th symmetric power for all $\lambda\in \mathbb{Z}_{\geq0}$.

Suppose $\lambda=m+1/2$, $m\in \mathbb{Z}_{\geq 0}$. For $N$ odd, we see from (\ref{9b}) that $C$ is the tensor product of its $N=1$ univariate form, with the connection matrix corresponding to the $(N-1)/2$-th
 symmetric power as specified above but with doubled parameters. For $N$ even, the connection matrix and thus the monodromy is
reducible. We have that $I_1^{(\alpha)}(x),I_3^{(\alpha)}(x),\ldots,I_{N-1}^{(\alpha)}(x)$ span an invariant subspace on which the monodromy acts as the $(N/2-1)$-th symmeric power of the univariate monodromy with double parameters, while on the quotient the monodromy acts as the $N/2$-th symmetric power of the univariate monodromy. As a consequence the scalar differential operator will then factorize. For example, with $\lambda=1/2$, $N=2$ we can check that the third order scalar differential operator can be written
\begin{equation*}
\hat{D}\Big(\frac{d}{dx}-\frac{\lambda_1+\alpha}{x}-\frac{\lambda_2+\alpha}{x-1}\Big)
\end{equation*}
where $\hat{D}$ is a second order operator. Here the first order operator has the solution
\begin{align}
x^{\lambda_1+\alpha}(1-x)^{\lambda_2+\alpha}&\propto \int_{0}^x dt_1 t_1^{\lambda_1}(1-t_1)^{\lambda_2}(x-t)^{\alpha-1}\int_x^1 dt_2 t_2^{\lambda_1}(1-t_2)^{\lambda_2}(t_2-x)^{\alpha-1}|t_1-t_2| \notag\\
&\propto I_1^{(\alpha)}(x),\label{xiI}
\end{align}
where the first proportionality follows from the Dixon-Anderson integral \cite[eq. (4.15)]{Fo10}. We remark that factorizations of differential operators play a prominent role in the recent studies of correlation functions in the two-dimensional Ising model \cite{ZBHM04,ZBHM05,BHMMWZ06,BHMMOZ07,
BBGHJMZ09}.

More generally, we see from (\ref{9b}) that for $\lambda=r/s$, $r$ and $s$ relatively prime positive integers. the elements of the connection matrix $C$ exhibit a congruence property. For $N\geq s$, $n\not=s+1$ we see that the latter is either a tensor product or block diagonal. As a consequence, in these cases $\{I_p^{(\alpha)}(x) \}$ can be written as powers or sums of solutions of lower order equations, as with (\ref{xiI}). In contrast, for generic $\lambda$ irrational the monodromy is irreducible and such reductions are not possible. 

To see this latter feature we use the fact \cite{An01} that the closure of any generic monodromy
group contains a conjugate of the monodromy group of any Fuchsian specialization. In
particular, it has been noted above that when $\lambda = 0$ the monodromy group is the
$N$-th symmetric power of the case $N = 2$. This in turn generically intersects
$SL_2$ in a Zariski dense subgroup. Thus the Zariski closure of the generic monodromy
group contains a copy of $SL_2$ in its (irreducible) $(N+1)$-dimensional representation.
In fact this subgroup of $SL_{N+1}$ is very nearly maximized in that the only permitted proper
subgroups of $SL_2$ strictly containing it are the orthogonal and symplectic groups. But our
monodromy matrices generically do not have eigenvalues occuring in complex
conjugate pairs and so our monodromy group is generically irreducible.

\subsection{Computation of $p_N(n;x)$}

With the joint eigenvalue PDF specified by (\ref{S}) we  have denoted by $p_N(n;x)$ the distribution of the $(n+1)$-st smallest eigenvalue $(n=0,1,\ldots,N-1)$. According to (\ref{E}), (\ref{2}) and (\ref{ea}) this distribution is given in terms of $\{I_k^{(\alpha)}(x) \}_{k=0,\ldots,n}$ by
\begin{equation}\label{pxn}
p_N(n;x)=-\frac{1}{S_N(\lambda_1,\lambda_2,\lambda)}\sum_{k=0}^n\frac{d}{d x} I_k^{(1)}(x). 
\end{equation}
We can use (\ref{ISC}) to compute the power series expansion of the $I_{k}^{(1)}(x)$, up to some given order, and thus the corresponding power series of $p_N(n;x)$.

Even with an infinite number of terms the pole at $x=1$ in the differential equation (\ref{ee}) tells us that generically each $I_k^{(1)}(x)$ will only have a radius of convergence of unity when expanded about the origin. Truncating the series must necessarily reduce this value as a practical figure for accurately determining (\ref{pxn}). But in keeping with (\ref{Hql}) we have the symmetry
\begin{equation*}
p_N(n;1-x)=p_N(N-n;x)\bigg|_{\lambda_1 \leftrightarrow \lambda_2},
\end{equation*}
implying that it is only necessary to calculate enough terms to ensure an accurate evaluation for $|x|\leq \frac{1}{2}$. 

Another relevant point is that $\sigma_k$ in (\ref{44'}) is, according to (\ref{19a}), a quadratic function of $k$. This means that a contribution of the linearly independent solution $(\mathbf{w}_N(x))_0^*$, say, does not show itself until order $x^{\sigma_N}=x^{O(N^2)}$ in the power series.

We have implemented the recurrence (\ref{89}) to compute the power series (\ref{44'}), and thus the power series of (\ref{ISC}) and (\ref{pxn}) using computer algebra software. This allows a case of (\ref{89}) for which $((\sigma+l)\mathbb{I}-A_0)$ is not invertible to be computed as the limit of a case of (\ref{89}) for which $((\sigma+l)\mathbb{I}-A_0)$ is invertible. But it comes at the expense of efficiency, and restricts us to small values of $N$. Nonetheless some interesting theory can be exhibited.

\begin{figure}[t]
\begin{center}
\includegraphics[height=5.5cm]{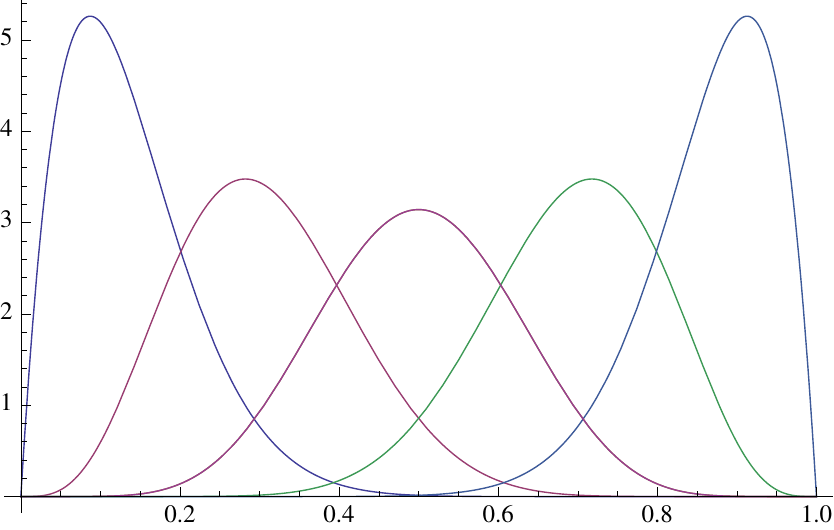}
\end{center}
\caption{Plot of the eigenvalue distributions $p_5(n;x)$ ($n=0,\dots,4$ reading left to right)
in the case $\lambda = 1/3$, $\lambda_1 = \lambda_2 = 1$. Due to $\lambda_1$ and
 $\lambda_2$ being equal we have the symmetry $p_5(n;x) = p_5(4-n;x)$, $n=0,1,2$. Also,
 $p_5(3;x)$ is given by (\ref{DM1}) with $a=b=1$.}
 \end{figure}

Generalizing results from \cite{FR01}, it was shown in \cite{Fo09} that
\begin{align}
&D_{r+1}({\rm ME}_{2/(r+1),(r+1)N+r}(x^a(1-x)^b)) \notag \\ 
&\hspace{1cm}={\rm ME}_{2(r+1),N}(x^{(r+1)a+2r}(1-x)^{(r+1)b+2r}). \label{DM}
\end{align}
Here ${\rm ME}_{\beta,m}(g(x))$ refers to the PDF proportional to 
\begin{equation*}
\prod_{l=1}^mg(x_l)\prod_{1\leq j < k \leq m}|x_k-x_j|^\beta
\end{equation*}
and $D_{r+1}(\cdot)$ denotes the distribution of the coordinates, assumed ordered $x_1<\ldots<x_m$, labelled by multiples of $(r+1)$ only. In particular, choosing $r=2$ and $N=1$, (\ref{DM}) tells us that 
\begin{equation}\label{DM1}
p_5(3;x)\bigg|_{\lambda=1/3}=\frac{1}{S_1(3a+4,3b+4,\cdot)}x^{3a+4}(1-x)^{3b+4}.
\end{equation} 
Indeed we find that computation of the series of $p_5(3;x)\big|_{\lambda=1/3}$ according to the above specified scheme is consistent with this result. A graphical illustration of our computation is given in Figure 1.

\section{Properties of the average  (\ref{8.1})}

\subsection{The case $\mu\in \mathbb{Z}^+$}
The average (\ref{8.1}) with $\mu\in \mathbb{Z}^+$, or more generally the average
\begin{equation} \label{xm}
\Big\langle \prod_{l=1}^N(x-t_l)^\nu\Big\rangle,\hspace{2cm}\nu\in \mathbb{Z}^+
\end{equation}
has the special property of being a polynomial in $x$. Moreover, it is simply related to the integral $I_N^{(\alpha)}(x)$. To see this make the change of variables $t_l \mapsto x t_l$ $(l=1,\ldots,N)$ in the definition (\ref{ea}) to obtain
\begin{equation*}
I_N^{(\alpha)}(x)=x^{\sigma_N}S_N(\lambda_1,\alpha-1,\lambda)\Big\langle \prod_{l=1}^N(1-x t_l)^{\lambda_2}\Big\rangle^\#
\end{equation*}
where the average $\langle \cdot \rangle^\#$ is with respect to the Selberg denisity $S_N(\lambda_1,\alpha-1,\lambda;\mathbf{t}).$ We read off from this that 
\begin{align}
\Big\langle \prod_{l=1}^N(x-t_l)^\nu\Big\rangle= &  \frac{x^{\nu N}}{S_N(\lambda_1,\lambda_2,\lambda)}\Big (x^{\sigma_N}I_N^{(\alpha)}(1/x) \Big )\bigg |_{\substack{\lambda_2\mapsto \nu \\ \alpha\mapsto \lambda_2+1}} \notag \\
= &  x^{\nu N}\Big (x^{\sigma_N} (\mathbf{w}_N ( \frac{1}{x} ))_0^* \Big )\Big |_{\substack{\lambda_2\mapsto \nu \\ \alpha\mapsto \lambda_2+1}} \label{IW}
\end{align}
where the second equality follows from (\ref{ISC}).

To gain intuition into the polynomial (\ref{xm}) it is useful to revise some known theory. First, a result of Aomoto \cite{Ao88} tells us that 
\begin{equation} \label{JP}
\Big\langle \prod_{l=1}^N(x-t_l)\Big\rangle=c_NP_N^{(\alpha,\beta)}(1-2x)
\end{equation}
where $c_N$ is chosen so that the coefficient of $x$ on the RHS has coefficient unity and 
\begin{equation} \label{JPa}
\alpha=(\lambda_1+1)/\lambda-1, \hspace{2cm} \beta=(\lambda_1+1)/\lambda-1
\end{equation}
and $P_N^{(\alpha,\beta)}(x)$ denotes the classical Jacobi polynomial. This latter interpretation draws attention to the asymptotic $N\rightarrow \infty$ behaviour of (\ref{xm}). As shown in \cite{BF97a} this task is made possible by the existence of a duality formula, refined in \cite[eq. (13.72)]{Fo10} to read
\begin{align}
\Big\langle \prod_{l=1}^N(x-t_l)^\nu \rangle \propto&  \Big\langle \prod_{l=1}^\nu e^{i \theta_l(\lambda_1-\lambda_2)/2\lambda} |1+e^{i \theta_l}|^{(\lambda_1+\lambda_2+2)/\lambda+N-2} \notag \\
& \hspace{3cm} \times \Big((1-x)e^{-i \theta_l/2}-xe^{i \theta_l/2} \Big)^N \Big\rangle_{{\rm CE}_{2/\lambda,\nu}}, \label{JP1}
\end{align}
where the constant of proportionality is determined as in (\ref{JP}) and $\langle \cdot \rangle _{{\rm CE}_{\beta,N}}$ denotes an average with respect to the circular $\beta$-ensemble
\begin{equation*}
\frac{1}{{\rm C}_{\beta,N}}\prod_{1\leq j < k\leq N}|e^{i\theta_k}-e^{i \theta_j}|^{\beta}.
\end{equation*}
The name duality formula comes from the interchange of the roles of the parameters $\nu$ and $N$ on the respective sides of (\ref{JP1}); also worth highlighting is the transformation of the coupling $\lambda\mapsto 2/\lambda$. These features characterize a wide class of duality formulas in random matrix theory \cite{De08, Ma08a}.

We can use (\ref{IW}) to give a very efficient evaluation of (\ref{xm}). Specifically, we use (\ref{26a}) to compute the coefficients of $\mathbf{p}_{l,N}$ $(l=1,\ldots, \nu N)$ in (\ref{44'}). According to (\ref{1q}), for positive parameters $\sigma_N>\sigma_j$ $(j=1,\ldots,N-1)$ and so $((\sigma_N+l)\mathbb{I}-A_0)$ in (\ref{26a}) is always invertible. Consequently the computation of $\mathbf{p}_{l,N}$ requires $l$ numerical matrix multiplications of $N\times N$ matrices and $l$ numerical matrix times vector operations at each step. Previously \cite{Fo93} the differential-difference equation (\ref{ec}) has been applied to this same task, as has an order $N+1$ difference equation satisfied by (\ref{ea}) for shifts $\alpha\mapsto \alpha+1$ \cite{FI10}. The drawback of these latter two approaches relative to the present one is the need to compute not just $I_N^{(\nu)}(x)$ but also $I_N^{(k)}(x)$ for each $k=1,\ldots, \nu-1$.

\begin{figure}[t]
\begin{center}
\includegraphics[height=6.5cm]{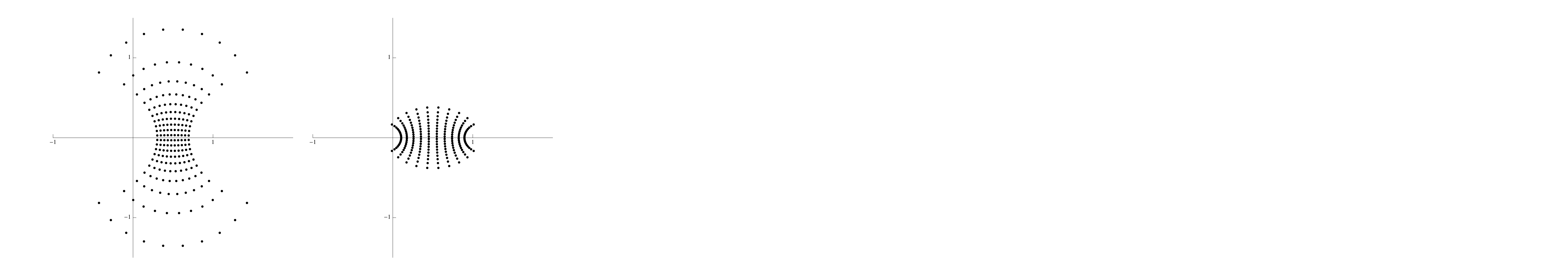}
\end{center}
\caption{Plot of the zeros of the polynomials (\ref{xm}) in the cases $N=10$, $\lambda = 1/3$,
$\lambda_1 = \lambda_2 = \lambda a$, $a=3$, $\nu = 20$ (left plot) and the same parameters 
$N,a,\nu$ but with
$\lambda = 3$. The ticks on the axes are the points $\pm 1$.\label{fig2}}
 \end{figure}
 
A consequence of our ability to efficiently compute the polynomial (\ref{xm}) is that we can investigate its zeros. As motivation, we recall that in the special case $\lambda=1$ (\ref{xm}) corresponds to a family of polynomial $\tau$-functions for the Painlev\`{e} $\rm{VI}$ system \cite{FW04}. The zeros of the polynomial $\tau$-functions for various Painlev\`{e} systems have been the subject of a number of studies by Clarkson (see e.g. \cite{Cl03a,Cl03b,Cl06}). Most interestingly, in the case of Painlev\`{e} $\rm{IV}$ it was shown that the zeros coincided with the equilibrium points for certain vortex flows \cite{Cl09}.

The free parameter $\lambda$ relative to the requirement $\lambda=1$ for the interpretation as a Painlev\`{e} $\tau$-function gives certain insights into the zeros of (\ref{xm}) for general $\lambda>0$. The main point is that with 
\begin{equation*}
\lambda_1\mapsto \lambda a, \hspace{2cm} \lambda_2\mapsto \lambda b, \hspace{2cm} \lambda\rightarrow \infty
\end{equation*}
the Selberg density crystallizes to the zeros of the Jacobi polynomial
\begin{equation} \label{Pn}
P_N^{(a-1,b-1)}(1-2x)
\end{equation}
(see e.g. \cite[Exercises 3.6 q.5]{Fo10}). Consequently in this limit (\ref{xm}) is then proportional to (\ref{Pn}) raised to the power of $\nu$. Remarkably (\ref{JP}), (\ref{JPa}) tells us that in the case $\nu=1$ this result persists for finite $\lambda$, up to the explicit value of the parameters $a,b$ in (\ref{Pn}). For $\nu>1$ and $\lambda$ finite we would expect at the very least that the $\nu$-fold degeneracy of the zeros present in the limit $\lambda\rightarrow \infty$ would be broken, but nonetheless that as in the case $\nu=1$ for there to remain signatures of the location of the zeros of (\ref{Pn}) when using parameters $\lambda_1 = \lambda a$, $\lambda_2 = \lambda b$.

Computation of the zeros shows that indeed the $\nu$-fold degeneracy of each zero present for $\lambda\rightarrow\infty$ is broken for finite $\lambda$. Instead for each of the $N$ previously $\nu$-fold degenerate zeros there are now $N$ distinctive curves in the complex plane containing the $\nu$ zeros,
and for large $\lambda$ these curves cut the real axis close to the zeros of (\ref{Pn}).
Because the polynomial (\ref{xm}) has real coefficients, the curves are invariant under reflection in the real axis, and this in turn tells us that for $\nu$ odd each curve contains a zero on the real axis. The length of the curves decreases as $\lambda$ increases, in keeping with the degeneracy for $\lambda\rightarrow \infty$. There are also distinctive curves tracing groups of $N$ zeros reading from left to right in the complex plane. At a qualitative level there being two classes of curves is in keeping with the duality formula (\ref{JP1}). The specific plots given in Figures \ref{fig2} and \ref{fig3}
illustrate these general features.

\begin{figure}[t]
\begin{center}
\includegraphics[height=6.5cm]{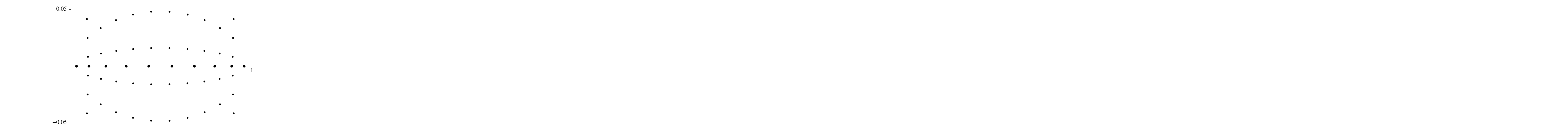}
\end{center}
\caption{Plot of the zeros of the polynomials (\ref{xm}) in the cases $N=10$, $\lambda = 3$,
$\lambda_1 = \lambda_2 = \lambda a$, $a=3$, $\nu = 20$ with imaginary part between
$-0.05$ and $0.05$ (zoom of right plot of Figure \ref{fig2}) super imposed with a plot of the
zeros of (\ref{Pn}) \label{fig3}}
 \end{figure}

\subsection{Asymptotics associated with (\ref{8.1})}

In (\ref{9.1}) and (\ref{9.2}) equations implying the functional forms of the asymptotic behaviour of the average (\ref{8.1}) in the limit $x\rightarrow 1^+$ are given. By the symmetry of the Selberg density (\ref{S}) under $t_l\mapsto1-t_l$ and $\lambda_1 \leftrightarrow \lambda_2$ this with $x\mapsto 1-x$ and $\lambda_2 \leftrightarrow \lambda_1$ gives the functional forms of the asymptotic behaviour of (\ref{8.1}) in the limit $x\rightarrow0^-$. We can use the expansion (\ref{ISC}) to reclaim these latter behaviours, and furthermore to specify the previously unknown proportionalities.

The first step is to write the average (\ref{8.1}) in terms of the integrals (\ref{ea}). Noting
\begin{equation*}
\int_{[0,1]^N}dt_1\ldots dt_N f(t_1,\ldots,t_N) =\sum_{q=0}^{N}\binom{N}{q}\int_{[0,x]^q}dt_1\ldots dt_q \int_{[x,1]^{N-q}} dt_{q+1}\ldots dt_N f(t_1,\ldots,t_N),
\end{equation*}
valid for any $f$ integrable on $[0,1]^N$ and symmetric, it follows
\begin{equation*}
\Big\langle \prod_{j=1}^N |t_j-x|^{2\mu} \Big\rangle = \frac{1}{S_N(\lambda_1,\lambda_2,\lambda)}\sum_{q=0}^{N}I_q^{(\alpha)}(x)\bigg|_{\alpha=2\mu+1} .
\end{equation*}
We now substitute (\ref{ISC}) to obtain
\begin{align}
&\Big\langle \prod_{j=1}^N|t_j-x|^{2\mu}\Big\rangle=\frac{1}{S_N(\lambda_1,\lambda_2,\lambda)}\sum_{q=0}^N\sum_{k=q}^N (\mathbf{w}_k(x))_0^*\binom{N}{k} \notag \\
& \hspace{1cm} \times S_k(\lambda_1,\alpha-1,\lambda)S_{N-k}(\lambda_1+\alpha-1+2k \lambda,\lambda_2,\lambda)c_{k,q}\bigg|_{\alpha=2\mu+1}. \label{SSc} 
\end{align}
By combining the inequalities (\ref{ss}) and (\ref{ts}) known from the proof of Proposition \ref{prop2} with the leading form of $(\mathbf{w}_k(x))^*$ as apparent from (\ref{51}), the following refinement of the asymptotic behaviour implied by (\ref{9.2}) is obtained.
\begin{proposition}\label{prop3}
Suppose $\lambda,\lambda_1,\lambda_2 \in \mathbb{R}$ obey (\ref{ea1}) but relax the condition on $\alpha$. Rather, for a given $l\in \{0,\ldots, N\}$ suppose
\begin{equation} \label{Sa}
-2l\lambda-\lambda_1<\alpha<-2(l-1)\lambda-\lambda_1, \hspace{1cm} \alpha=2\mu+1.
\end{equation}
We have
\begin{align}
&\Big\langle \prod_{j=1}^N|t_j-x|^{2\mu} \Big\rangle \mathop{\sim}_{x\rightarrow0}x^{\sigma_l}\frac{S_l(\lambda_1,\alpha-1,\lambda)S_{N-l}(\lambda_1+\alpha-1+2l\lambda,\lambda_2,\lambda)}{S_N(\lambda_1,\lambda_2,\lambda)} \notag \\
&\hspace{3.5cm} \times \binom{N}{l}\sum_{q=0}^{l}c_{l,q}\bigg|_{\alpha=2\mu+1}. \label{Sa1}
\end{align}
\end{proposition}

At the left endpoint of (\ref{Sa}) we have $-2l\lambda-\lambda_1=\alpha$. According to (\ref{ts}) the minimum of $\sigma(x)$ then occurs at $x=l+1/2$. This means that both $\sigma_l$ and $\sigma_{l+1}$ are equally the smallest exponents and consequently
\begin{equation} \label{Sa2}
\Big\langle \prod_{j=1}^N |t_j-x|^{2\mu}\Big\rangle \mathop{\sim}_{x\rightarrow 0} \lim_{\alpha\rightarrow -2l\lambda-\lambda_1}\big(x^{\sigma_l}h(l)+x^{\sigma_{l+1}}h(l+1)\big)
\end{equation}
where $x^{\sigma_l}h(l)$ refers to the RHS of (\ref{Sa1}). The limit in (\ref{Sa2}) can be computed to obtain a refinement of the asymptotic behaviour implied by (\ref{9.1}).

\begin{proposition}
Suppose $\lambda,\lambda_1,\lambda_2 \in \mathbb{R}$ obey (\ref{ea1}) but relax the condition on $\alpha$ and replace it by
\begin{equation} \label{Sa3}
\alpha=-2l\lambda-\lambda_1, \hspace{1cm} l\in\{0,\ldots,N-1\}.
\end{equation}
We have
\begin{align}
\Big\langle \prod_{j=1}^N|t_j-x|^{2\mu} \Big\rangle \mathop{\sim}_{x\rightarrow0}&\Big(|x|^{\sigma_l}  \log \frac{1}{|x|} \Big)\frac{S_l(\lambda_1,\alpha-1,\lambda)S_{N-l}(\lambda_1+\alpha-1+2l\lambda,\lambda_2,\lambda)}{S_N(\lambda_1,\lambda_2,\lambda)\Gamma(\lambda_1+\alpha+2l \lambda)} \notag \\
& \times \binom{N}{l}\sum_{q=0}^{l}c_{l,q}\bigg|_{\alpha=2\mu+1}.\label{Sa4}
\end{align}

\end{proposition}

\begin{proof}
We seek the individual leading order asymptotic form of the two terms in (\ref{Sa2}). For the first of these, by definition 
\begin{equation} \label{LS}
h(l)=\binom{N}{l}\frac{S_l(\lambda,\alpha-1,\lambda)}{S_N(\lambda_1,\lambda_2,\lambda)}S_{N-l}(\lambda_1+\alpha-1+2l\lambda,\lambda_2,\lambda)\sum_{q=0}^Nc_{l,q}.
\end{equation}
Now we see from (\ref{2a}) that the term $S_{N-l}$ in (\ref{LS}) contains factors
\begin{equation*}
\Gamma(\lambda_1+\alpha+(2l+j)\lambda)\hspace{2cm} (j=0,\ldots,N-l-1)
\end{equation*}
With $j=0$ we see that this factor diverges as $1/(\alpha+\lambda_1+2l\lambda)$ while all other terms in (\ref{LS}) remain finite and so 
\begin{align}
h(l)\mathop{\sim}_{\alpha \rightarrow -2l \lambda-\lambda_1}\hspace{0.2cm}&\frac{1}{\alpha+\lambda_1+2l \lambda}\binom{N}{l}\frac{S_l(\lambda_1,\alpha-1,\lambda)}{S_N(\lambda_1,\lambda_2,\lambda)}\frac{S_{N-l}(\lambda_1+\alpha-1+2l\lambda,\lambda_2,\lambda)}{\Gamma(\lambda_1+\alpha+2l\lambda)} \notag \\
&\times \sum_{q=0}^lc_{l,q}. \label{h0}
\end{align}

In contrast to the behaviour of $h(l)$ in the limit $\alpha\rightarrow-2l\lambda-\lambda_1$, we see that $h(l+1)$, defined by (\ref{LS}) with $l\mapsto l+1$ has all terms involving Selberg integrals remaining finite in this limit. Instead, as seen from (\ref{C2}), each term $c_{l,q}$ $(q=0,\ldots,l)$ diverges (but not $c_{l+1,l+1}=1$) and we have 
\begin{align}
 h(l+1) \mathop{\sim}_{\alpha\rightarrow-2l\lambda-\lambda_1} &\binom{N}{l+1}     \frac{S_{l+1}(\lambda_1,\alpha-1,\lambda)}{S_N({\lambda_1,\lambda_2,\lambda)}}S_{N-l-1}(\lambda_1+\alpha-1+2(l+1)\lambda,\lambda_2,\lambda)  \notag \\
& \times c_{l+1,l} \sum_{q=0}^{l} \frac{c_{l+1,q}}{c_{l+1,l}}\label{h51}
\end{align}
The expressions (\ref{h51}) and (\ref{h0}) are closely related. First, we note from (\ref{C2}) that in the limit $\alpha\rightarrow -2l\lambda-\lambda_1$,
\begin{equation}\label{C1b}
c_{l+1,l}\mathop{\sim}\frac{\sin \pi(l+1)\lambda \sin\pi(\lambda_1+l\lambda)}{\pi (\lambda_1+\alpha+2l \lambda) \sin \pi \lambda}
\end{equation}
and 
\begin{equation}\label{C2b}
\sum_{q=0}^l\frac{c_{l+1,q}}{c_{l+1,l}}\mathop{\sim}\sum_{q=0}^{l}c_{l,q}.
\end{equation}
Then we use (\ref{2a}), and the reflection formula for the gamma function to deduce that in the same limit
\begin{equation}\label{C3b}
\frac{S_{l+1}(\lambda_1,\alpha-1,\lambda)}{S_l(\lambda_1,\alpha-1,\lambda)}\mathop{\sim}-\frac{\pi (l+1)\sin\pi\lambda}{\sin \pi (l+1)\lambda \sin \pi (\lambda_1+l \lambda)}
\end{equation}
and
\begin{equation}\label{C4b}
\frac{S_{N-l-1}(\lambda_1+\alpha-1+2(l+1)\lambda,\lambda_2,\lambda)\Gamma(\lambda_1+\alpha+2l\lambda)}{S_{N-l}(\lambda_1+\alpha-1+2l\lambda,\lambda_2,\lambda)}\mathop{\sim}\frac{1}{N-l}.
\end{equation}
Substituting (\ref{C1b})-(\ref{C4b}) in (\ref{h51}) shows
\begin{equation}\label{h1}
h(l+1)\mathop{\sim}_{\alpha\rightarrow -2l \lambda-\lambda_1}-h(l)
\end{equation}
Now substituting (\ref{h0}) and (\ref{h1}) in (\ref{Sa2}) implies the result (\ref{Sa4}). 
\hfill $\square$  \end{proof}

\section*{Acknowledgement}
The contribution to the preparation of this paper by Wendy Baratta and James Saunderson is
acknowledged. This work was supported by the Australian Research Council.


\providecommand{\bysame}{\leavevmode\hbox to3em{\hrulefill}\thinspace}
\providecommand{\MR}{\relax\ifhmode\unskip\space\fi MR }
\providecommand{\MRhref}[2]{%
  \href{http://www.ams.org/mathscinet-getitem?mr=#1}{#2}
}
\providecommand{\href}[2]{#2}

\end{document}